\documentclass[conference]{IEEEtran}
\IEEEoverridecommandlockouts

\usepackage{cite}
\usepackage{soul}
\usepackage{amsmath,amssymb,amsfonts}
\usepackage{algorithm}
\usepackage{algorithmic}
\usepackage{graphicx}
\usepackage{textcomp}
\usepackage{xcolor}
\usepackage{amssymb}
\usepackage{amsfonts}
\usepackage{hyperref}
\usepackage{cleveref}
\usepackage{dsfont}
\usepackage{amsthm}
\usepackage{url}
\usepackage{graphicx} 
\newtheorem{theorem}{Theorem}
\newtheorem{definition}{Definition}

\newtheorem{remark}{Remark}
\newtheorem{lemma}{Lemma}
\newtheorem{corollary}{Corollary}

\def\BibTeX{{\rm B\kern-.05em{\sc i\kern-.025em b}\kern-.08em
    T\kern-.1667em\lower.7ex\hbox{E}\kern-.125emX}}
    
\begin{document}

\title{System Level Synthesis for Affine Control Policies: Model Based and Data-Driven Settings

}

\author{
    Lukas Sch\"uepp  \and
    Giulia De Pasquale \and Florian D\"orfler  \and Carmen Amo Alonso \thanks{Lukas Sch\"uepp and Florian D\"orfler are with the Automatic Control Laboratory, Department of Electrical Engineering and Information Technology, ETH Zurich, Physikstrasse
3 8092 Zurich, Switzerland.(e-mail:\{ lukaschu,dorfler\}@ethz.ch). Carmen Amo Alonso is with the Computer Science Department of Stanford University, Palo Alto, CA. Giulia De Pasquale is with with the Control Systems Group, Department of Electrical Engineering, TU Eindhoven, Flux
Groene Loper 19
5612AP Eindhoven, The 
Netherlands. (e-mail: g.de.pasquale@tue.nl). This research is supported by the Swiss National Science Foundation under NCCR Automation.}}
\maketitle
\begin{abstract}
There is an increasing need for effective control of systems with complex dynamics, particularly through data-driven approaches.
System Level Synthesis (SLS) has emerged as a powerful framework that facilitates the control of large-scale systems while accounting for model uncertainties. SLS approaches are currently limited to linear systems and time-varying linear control policies, thus limiting the class of achievable control strategies.
We introduce a novel closed-loop parameterization for time-varying affine control policies, extending the SLS framework to a broader class of systems and policies. We show that the closed-loop behavior under affine policies can be equivalently characterized using past system trajectories, enabling a fully data-driven formulation. This parameterization seamlessly integrates affine policies into optimal control problems, allowing for a closed-loop formulation of general Model Predictive Control (MPC) problems. To the best of our knowledge, this is the first work to extend SLS to affine policies in both model-based and data-driven settings, enabling an equivalent formulation of MPC problems using closed-loop maps. We validate our approach through numerical experiments, demonstrating that our model-based and data-driven affine SLS formulations achieve performance on par with traditional model-based MPC.
\end{abstract}


\section{Introduction}
The increasing complexity of control systems, ranging from power grids \cite{power_grid_ctrl} to social networks \cite{large_scale_sn}, has challenged the effectiveness of traditional control paradigms. While Model Predictive Control (MPC) is often an effective strategy, it can be limited by the complex nature of these systems.  
The overall dynamics turn out  highly non-linear and the intricate interactions among system states make the identification of precise analytical models a complex task. 
In this context,
data-driven approaches have become popular as they eliminate the need for explicitly defining a model of the system, reducing the burden of manual identification \cite{dd_large_scale_opt,dd_large_scale_loewner,data_dlmpc}.  
However, solutions beyond the linear policy and linear systems domain are an active research topic as in that case controlling a system with only past trajectory information becomes a significantly harder problem \cite{affine_sys_ctrl} \cite{Martinelli_2022}.

\textit{Prior work.} The System Level Synthesis (SLS) \cite{sls} framework provides a parameteirzation of closed loop system responses that shift the focus of the control design from the controller to the whole closed-loop system. This paradigm shift has enabled significant advances in distributed MPC \cite{dlmpc} \cite{data_dlmpc} and robust MPC \cite{sls_robust, sls_antoine, sieber_sls, chen2020robust}. Other applications span infinite horizon optimal control \cite{inf_horizon_sls} and input-output formulations \cite{input_output_sls}.
However, the traditional SLS approach restricts itself to linear systems, which substantially narrows its applicability. Though efforts have been made to extend SLS beyond linear systems \cite{NN_SLS, Nonlin_sls_pol_dyn, sls_appr_for_DTL}, these extensions were limited to unconstrained optimal control problems.
A fundamental limitation of all these existing SLS-based methods is their reliance on purely linear control policies. As demonstrated in \cite{borrelli_constrained}, linear policies lack the expressivity needed to capture the piecewise affine nature of solutions that inherently arise in MPC problems, thus highlighting the necessity to extend closed-loop parameterizations to affine policies.
Furthermore, given the abundance of data and the intricacies of system modeling, data-driven methods have gained significant traction in control theory \cite{model_to_data_driven_shift}. More specifically for MPC, several innovations have emerged for both linear \cite{Berberich_dd_MPC, ber_köhler_rob_dd_MPC} and nonlinear systems \cite{Johannes_dd_MPC}. While tube-based approaches offer stability and robustness guarantees through constraint tightening, they typically require offline tube computation, introducing inherent conservatism \cite{chen2020robust}. In contrast, SLS-based robust methods directly design closed-loop maps from disturbances to system trajectories, providing a more transparent characterization of disturbance effects \cite{chen2020robust}. Though the SLS framework has been extended to data-driven settings \cite{data_sls} and applied effectively for robust control \cite{micheli2024data}, it has yet to accommodate affine policies or address the broader class of affine systems. Despite recent advances in data-driven control for affine systems \cite{Martinelli_2022, affine_sys_ctrl}, integrating these approaches with closed-loop formulations under affine policies—particularly for constrained control problems—remains an open challenge that our work addresses.

\textit{Contribution.}
Our work addresses this gap by introducing a novel closed-loop parameterization with time-varying affine policies and extending it to the data-driven setting. The contributions of our approach are two-fold. On one hand, it provides a generalization of the SLS framework to affine policies. On the other hand, it enables a seamless integration into all nominal MPC formulations with quadratic cost, affine prediction models and polytopic constraints. We extend our approach to the data-driven setting and provide a data-driven characterization of system responses for time-invariant affine systems under time-varying affine policies. To validate our approach, we evaluate the affine SLS framework in both model-based and data-driven settings, comparing it against traditional MPC. Our results demonstrate that the proposed methods achieve the same performance as conventional model-based MPC, confirming their effectiveness.

\textit{Paper Structure.} The remainder of this paper is organized as follows: \Cref{sec:problem_form_section} presents the problem formulation. \Cref{sec:sls_affine_section} presents the SLS with affine policies in the model-based case. \Cref{sec:data_driven_appr} introduces the data-driven approach of the SLS parameterization with affine policies. \Cref{section_simulation} validates our approach via numerical experiments. \Cref{discussion_section} concludes the paper.

\textit{Notation.} 
Lower-case and upper-case Latin and Greek letters such as $x$ and $A$ denote vectors and matrices, respectively, though lower-case letters may also represent scalars or functions, with distinctions made clear by context.  For $t_1, t_2 \in \mathbb{N}$ with $t_1 < t_2$, both $[t_1,t_2]$ and $t_1:t_2$ represent the integer set $\{t_1,\dots, t_2 \}$.  Given a matrix $A \in \mathbb{R}^{m \times n}$, we denote its $(i,j)$ entry by $A_{i,j}$. Bracketed indices indicate the true system time, i.e., the system is at state $x(t)$ at time $t$. Subscripts denote prediction time indices within an MPC loop, i.e., $x_t$ is the $t^\mathrm{th}$ predicted state. Calligraphic letters, such as $\mathcal{S}$, represent sets or concatenated matrices -- the distinction will be clear from context. Boldface letters, such as $\mathbf{x}$ and $\mathbf{K}$, denote finite-horizon signals and lower block triangular, i.e., causal, operators respectively:
\[
\mathbf{x}=\begin{bmatrix} x_{0}\\x_{1}\\\vdots\\x_{T} \end{bmatrix},\quad
\mathbf{K}=\begin{bmatrix} K^{0,0} & & & \\ K^{1,1} & K^{1,0} & & \\ \vdots & \ddots & \ddots & \\ K^{T,T} & \dots & K^{T,1} & K^{T,0} \end{bmatrix},
\]
where each $K^{i,j}$ is a matrix of compatible dimension. Here, $\mathbf{K}$ represents the convolution operation induced by a time-varying controller $K_t(x_{0:t})$, such that $u_t = \mathbf{K}^{t,:} \mathbf{x}$, where $\mathbf{K}^{t,:}$ represents the $t$-th block-row of $\mathbf{K}$. Let $I_n \in \mathbb{R}^{n\times n}$ denote the identity matrix of size $n$ and $\mathbf{0}_{n\times m}$ the zero matrix of dimension $n\times m$. The vectors $\mathbf{1}_n$ and $\mathbf{0}_n$ represent the $n$-dimensional column vectors whose entries are all equal to $1$ and $0$, respectively. When clear from context, subscripts will be omitted. The symbol $\lVert \cdot \rVert_p$ denotes the $p$-norm.

\section{Problem Formulation}
\label{sec:problem_form_section}
Consider the discrete-time varying affine system with dynamics
\begin{equation}\label{eqn:system}
x(t+1) = A(t)x(t) + B(t)u(t) + s + w(t),
\end{equation}
where $x \in \mathbb{R}^{n}$ is the state, $u \in \mathbb{R}^{m}$ is the input, $s\in \mathbb{R}^n$ is a constant term, and $w \in \mathbb{R}^n$ is the disturbance or noise. Our goal is to define a closed-loop system response over a time horizon $T$ for system \eqref{eqn:system}, under time-varying affine control policies.
Furthermore, we study the design of a closed-loop MPC policy for system \eqref{eqn:system} in the noiseless case with time-invariant system matrices $A, B$. As is standard, a model predictive controller is implemented by solving a series of finite horizon optimal
control problems, with the problem solved at time $\tau$ parameterized with the initial condition $x_0 = x(\tau)$ over a prediction
horizon $T$ given by:
\begin{align}\label{Traditional_MPC}
    & \underset{{u_{0:T}, x_{0:T}}}{\min}&  & J(u_{0:T},x_{0:T}) \\
    & \text{s.t.} &  &\begin{aligned} \nonumber
    &x_{t+1} =  A  x_t +  B u_t + s\\
    & H_x x_t+ H_uu_t \leq h\\
    &x_0 = x(0),\text{ } x_T \in \mathcal{X}_T \quad t\in[0,T],
\end{aligned}
\end{align}
where the subscript $t$ refers to the variable's prediction $t$ steps ahead. The matrices $H_x\in \mathbb{R}^{d \times n}$, $H_u\in \mathbb{R}^{d \times m}$ and $h \in \mathbb{R}^d$, define $d$ polytopic constraints.
$\mathcal{X}_T$ is a polyhedra that defines the terminal set. The cost function $J$ is defined as
\begin{equation}\label{cost_fn}
J(u_{0:T},x_{0:T})= \lVert Px_T \rVert_p + \sum_{t=0}^{T-1}\lVert Qx_t\rVert_p + \lVert Ru_t \rVert_p,
\end{equation}
where matrices $Q \in \mathbb{R}^{n\times n}$, $R \in \mathbb{R}^{m \times m}, P\in \mathbb{R}^{n \times n}$ respectively and  $p\in\{1,2,\infty\}$. For $p=2$, we assume that $Q=Q^{\top}\succeq0 \text{, }R=R^{\top}\succ0\text{, } P=P^{\top}\succeq0$. If $p=1$ or $p=\infty$, we assume that $Q,P,R$ are full  rank matrices \cite{borrelli_constrained}.

It is a well-known fact that the control policies resulting from the MPC problem \eqref{Traditional_MPC} are piece-wise affine policies \cite{morari_MPQP}. However, closed-loop versions of the SLS formulation, in both the model-based \cite{sls} and the data-driven setting \cite{data_sls}, provide a \emph{linear} time-varying (LTV) parameterization of the resulting control policies \cite{dlmpc,data_sls}, where the predicted control action at time $t$, $u_t$, is a linear combination of all previously predicted states $x_0, x_1, \dots, x_{t_0}$. More specifically
\begin{equation} \label{linear_policies}
    u_t = \sum_{\tau =0}^t K_{t,\tau}x_\tau
\end{equation}
where $K_{t,\tau} \in \mathbb{R}^{m \times n}$ is a linear feedback law that maps the state $x_t$ onto the input $u_\tau$ with $\tau \leq t$. Therefore, the LTV policy \eqref{linear_policies} is not rich enough to encompass all possible policies solving \eqref{Traditional_MPC}.

Here, we study a closed-loop parameterization of system \eqref{eqn:system} with time-varying affine control policies, where we replace \eqref{linear_policies} by
\begin{equation} \label{ctrl_policy}
u_t = \sum_{\tau =0}^t K_{t,\tau} x_\tau + u_{t}^s,
\end{equation}
where $u_{t}^s \in \mathbb{R}^{m}$ is the affine part of the policy at time~$t$.

We leverage this parameterization to reformulate the MPC problem \eqref{Traditional_MPC} using closed-loop maps. We do this by first introducing a closed-loop parameterization of affine policies that allow to generalize previous MPC formulations \cite{dlmpc}. Then, we port the resulting parameterization into the data-driven case where only past trajectories are available. 

\section{System Level Synthesis with Affine Policies}\label{sec:sls_affine_section}
We introduce a closed-loop parameterization that extends the SLS parameterization introduced in \cite{sls} to accommodate time-varying affine control policies. We first show how this naturally permits the closed-loop parameterization of affine systems. Then, we illustrate how the new closed-loop system responses can be used to equivalently reformulate the MPC problem \eqref{Traditional_MPC} with closed-loop maps. In this section, we operate under the assumption that dynamic matrices $A$, $B$ and the constant $s$ are known. We extend our results to the data-driven setting in \ref{sec:data_driven_appr}.

\subsection{System Level Synthesis with Affine Policies}
We start by rearranging the system's matrices over time into block-diagonal matrices $\mathcal{A} = {\rm diag}(A(0),\dots,A(T\!-\!1),\mathbf{0}_{n \times n})$, $\mathcal{B} = {\rm diag}(B(0),\dots,B(T\!-\!1),\mathbf{0}_{n \times m})$. Moreover, we use the signal symbols to denote length-$(T\!+\!1)$ time series of the respective variables: $\mathbf{x}\in \mathbb{R}^{n(T+1)}$ represents the state trajectory, $\mathbf{u}\in \mathbb{R}^{m(T+1)}$ represents the input trajectory, $\mathbf{s}=[s^{\top},...,s^{\top},\mathbf{0}_{n}^\top]^{\top}\in \mathbb{R}^{n(T+1)}$ represents the constant setpoint repeated over the horizon, and $\mathbf{w}=[x(0)^{\top},w(0)^{\top},...,w(T\!-\!1)^{\top}]^{\top}\in \mathbb{R}^{n(T+1)}$ represents the initial state and disturbances. Similarly, the control policy \eqref{ctrl_policy} over a time horizon $T$ can be rewritten using a lower block diagonal operator $\mathbf{K}$ and a signal vector $\mathbf{u_s}$ as
\begin{equation} \label{eq:k_law}
\mathbf{u} = \mathbf{K}   \mathbf{x}  + \mathbf{u_s}. 
\end{equation}
Let $\mathcal{Z}$ be the block-downshift operator, a matrix with identity blocks along its first sub-diagonal and zeros elsewhere.
We rewrite the dynamics of \eqref{eqn:system} for a time horizon of length $T$ as
\begin{equation*}
    \mathbf{x} =\mathcal{Z}\mathcal{A} \mathbf{x} + \mathcal{Z}\mathcal{B}\left({\mathbf K}\mathbf{x}+ \mathbf{u_s}\right) + \mathcal{Z}\mathbf{s} + \mathbf{w}.
\end{equation*}
By rearranging the terms above one gets
\begin{equation*} 
    (I_{n(T+1)} -\mathcal{Z}\mathcal{A} - \mathcal{Z}\mathcal{B}\mathbf K)\mathbf{x} = \mathcal{Z}(\mathcal{B}\mathbf{u_s} + \mathbf{s}) + \mathbf{w},
\end{equation*}
where $(I_{n(T+1)}-\mathcal{Z}\mathcal{A}-\mathcal{Z}\mathcal{B} \mathbf K)$ is a lower block diagonal matrix with identity blocks on the diagonal. Since this is invertible, we can alternatively write the above as \begin{equation*}
\mathbf{x} = (I_{n(T+1)} -\mathcal{Z}\mathcal{A} - \mathcal{Z}\mathcal{B}\mathbf K)^{-1}[I_{n(T+1)} \quad \mathcal{Z}(\mathcal{B}\mathbf{u_s}+\mathbf{s})]\begin{bmatrix} \mathbf{w} \\ \mathbf{1} \end{bmatrix}.
\end{equation*}
In the above expression, it is possible to recover the state system response $\mathbf{\Phi_x}$ to the initial condition and disturbance from the traditional SLS parameterization \cite{sls}, which is a lower block-diagonal matrix with identities along the diagonal:
\begin{subequations} \label{eq:state_system_response}
    \begin{gather} \label{eqn:Phi_x}
        \boldsymbol{\Phi_x}=  (I_{n(T+1)} -\mathcal{Z}\mathcal{A} - \mathcal{Z}\mathcal{B} \mathbf K)^{-1}.
    \intertext{Furthermore, the remaining term can be seen as the complement to the affine system response $\boldsymbol{\phi_x}  \in \mathbb{R}^{n(T+1)}$:}
     \label{eqn:phi_x}
        \boldsymbol{\phi_x} = \boldsymbol{\Phi_x}\mathcal{Z}(\mathcal{B}\mathbf{u_s}+\mathbf{s}).
    \end{gather}
\end{subequations}
One can rewrite the expression for the affine controller \eqref{eq:k_law} as
\begin{equation*}
\begin{aligned}
\mathbf{u} &= \mathbf{K}\mathbf{x} + \mathbf{u_s} = \mathbf{K} [\boldsymbol{\Phi_x} \quad \boldsymbol{\phi_x}]\begin{bmatrix} \mathbf{w} \\ \mathbf{1} \end{bmatrix} + \mathbf{u_s}\\
&=\begin{bmatrix}\mathbf{K}\boldsymbol{\Phi_x} & \mathbf{K}\boldsymbol{\phi_x} + \mathbf{u_s} \end{bmatrix} \begin{bmatrix} \mathbf{w} \\ \mathbf{1} \end{bmatrix}.
\end{aligned}
\end{equation*}

Hence, the map between $\mathbf u$ and $\mathbf w$ is identical to the traditional SLS system response, i.e., $\boldsymbol \Phi_u= \mathbf{K} \boldsymbol \Phi_x$, which is a lower block-diagonal matrix:
\begin{subequations}\label{eq:input_system_response}
    \begin{gather}\label{eqn:Phi_u}
        \boldsymbol{\Phi_u} = \mathbf K(I_{n(T+1)} -\mathcal{Z}\mathcal{A} - \mathcal{Z}\mathcal{B} \mathbf K)^{-1}.
    \intertext{Furthermore, the remaining term $\mathbf{K}\boldsymbol{\phi_x} + \mathbf{u_s} = \mathbf{K}\boldsymbol{\Phi_x}\mathcal{Z}(\mathcal{B}\mathbf{u_s}+\mathbf{s})+ \mathbf{u_s}$ can be seen as the complement to the affine system response $\boldsymbol \phi_u\in \mathbb{R}^{m(T+1)}$}
    \label{eqn:phi_u}
        \boldsymbol{\phi_u} = \boldsymbol{\Phi_u} \mathcal{Z} (\mathcal{B}\mathbf{u_s}+\mathbf{s}) + \mathbf{u_s}.
    \end{gather}
\end{subequations}
These derivations motivate the following definition:
\begin{definition}\emph{(Affine SLS Parameterization).}
    Given the dynamics of system \eqref{eqn:system}, the closed-loop maps from the disturbance $\mathbf{w}$ to the state and input trajectories, $\mathbf{x}$ and $\mathbf{u}$ respectively, under affine policies \eqref{eq:k_law} are of the form
\begin{equation}\label{closed_loop_response}
\begin{bmatrix} \mathbf{x} \\ \mathbf{u}\end{bmatrix} = \begin{bmatrix} \boldsymbol{\Phi_x} & \boldsymbol{\phi_x} \\ \boldsymbol{\Phi_u} & \boldsymbol{\phi_u}\end{bmatrix} \begin{bmatrix} \mathbf{w} \\ \mathbf{1} \end{bmatrix},
\end{equation}
where $\{\boldsymbol{\Phi_x},\boldsymbol{\phi_x},\boldsymbol{\Phi_u},\boldsymbol{\phi_u}\}$ satisfy \eqref{eqn:Phi_x}-\eqref{eqn:phi_u}.
\end{definition}
In what follows, we show that the affine SLS parameterization \eqref{closed_loop_response} extends the linear SLS framework of \cite[Theorem 2.1]{sls} to accommodate affine control policies.
\begin{theorem}[SLS with Affine Policies]\label{SLS_formulation}
Given the system dynamics \eqref{eqn:system} with affine state feedback law \eqref{ctrl_policy} over the time horizon $t\in[0,T]$, the following are true:
\begin{enumerate} \label{theorem_sls}
    \item The affine subspace induced by\footnote{With abuse of notation $I=I_{n(T+1)}$.}
    \begin{equation} \label{sys_constraints}
    \begin{aligned}
    &\begin{bmatrix} I-\mathcal{Z}\mathcal{A} & -\mathcal{Z} \mathcal{B} \end{bmatrix}
    \begin{bmatrix} \boldsymbol{\Phi_x} & \boldsymbol{\phi_x} \\ \boldsymbol{\Phi_u} & \boldsymbol{\phi_u}\end{bmatrix} = \begin{bmatrix} I & \mathcal{Z}\mathbf{s}\end{bmatrix}
     \end{aligned}
    \end{equation}
    parameterizes all possible closed-loop maps \eqref{closed_loop_response}. That is, all system responses \eqref{eq:state_system_response} and \eqref{eq:input_system_response} satisfy the above equality.
    \item For any set $\{\boldsymbol{\Phi_x}, \boldsymbol{\phi_x},\boldsymbol{\Phi_u},\boldsymbol{\phi_u}\}$ satisfying \eqref{sys_constraints}, the control policy \eqref{eq:k_law} with $\mathbf{K}= \boldsymbol{\Phi_u} \boldsymbol{\Phi_x} ^{-1}$ and $\mathbf{u_s} = \boldsymbol{\phi_u} - \boldsymbol{\Phi_u\Phi_x}^{-1}\boldsymbol{\phi_x}$ achieves the desired closed-loop maps \eqref{closed_loop_response} with system responses \eqref{eq:state_system_response} and \eqref{eq:input_system_response}.
\end{enumerate}
\end{theorem}
\begin{proof} Statement $1)$ Let us separate equation \eqref{sys_constraints} into two equalities. The first one reads as
\begin{equation*}
\begin{bmatrix} I-\mathcal{Z}\mathcal{A} -\mathcal{Z}\mathcal{B}\end{bmatrix} \begin{bmatrix} \boldsymbol{\Phi_x} \\ \boldsymbol{\Phi_u}\end{bmatrix} = I.
\end{equation*}
This subspace constraint is identical to the standard SLS parameterization \cite{sls}. We present the derivation here for completeness. Specifically, using the definition of $\boldsymbol{\Phi_x}$ and $\boldsymbol{\Phi_u}$ from \eqref{eqn:Phi_x} and \eqref{eqn:Phi_u}, respectively, we have that:
\begin{equation*}
\begin{aligned}
&(I-\mathcal{Z}\mathcal{A})(I-\mathcal{Z}\mathcal{A}-\mathcal{Z}\mathcal{B}\mathbf{K})^{-1} -\mathcal{Z}\mathcal{B}\mathbf{K}(I-\mathcal{Z}\mathcal{A}-\mathcal{Z}\mathcal{B}\mathbf{K})^{-1}\\
& =(I-\mathcal{Z}\mathcal{A} - \mathcal{Z}\mathcal{B}\mathbf{K})(I-\mathcal{Z}\mathcal{A} - \mathcal{Z}\mathcal{B}\mathbf{K})^{-1}= I.\\
\end{aligned}
\end{equation*}

The second equality reads as
\begin{equation*}
\begin{bmatrix} I-\mathcal{Z}\mathcal{A} &-\mathcal{Z}\mathcal{B}\end{bmatrix} \begin{bmatrix} \boldsymbol{\phi_x} \\ \boldsymbol{\phi_u}\end{bmatrix} = \mathcal{Z}\mathbf{s}.
\end{equation*}
We show that this equality holds with the definitions of $\boldsymbol{\phi_x}$ and $\boldsymbol{\phi_u}$ from \eqref{eqn:phi_x} and \eqref{eqn:phi_u}, respectively.  We can write
\begin{equation*}
\begin{aligned}
&(I-\mathcal{Z}\mathcal{A})\boldsymbol{\Phi_x}\mathcal{Z}(\mathcal{B}\mathbf{u_s}+\mathbf{s}) - \mathcal{Z}\mathcal{B}(\boldsymbol{\Phi_u}\mathcal{Z}(\mathcal{B}\mathbf{u_s}+\mathbf{s}) + \mathbf{u_s})\\
&=\begin{bmatrix}I-\mathcal{Z}\mathcal{A} & -\mathcal{Z}\mathcal{B} \end{bmatrix}
\begin{bmatrix}\boldsymbol{\Phi_x}\\ \boldsymbol{\Phi_u}\end{bmatrix} \mathcal{Z}(\mathcal{B}\mathbf{u_s} + \mathbf{s}) -\mathcal{Z}\mathcal{B}\mathbf{u_s}\\
&= \mathcal{Z}(\mathcal{B}\mathbf{u_s} + \mathbf{s}) -\mathcal{Z}\mathcal{B}\mathbf{u_s} = \mathcal{Z}\mathbf{s},
\end{aligned}
\end{equation*}
where the second equality holds because, as previously shown, $(I-\mathcal{Z}\mathcal{A})\boldsymbol{\Phi_x} - \mathcal{Z}\mathcal{B}\boldsymbol{\Phi_u} = I$.

Statement $2)$ First we note that $\boldsymbol{\Phi_x}$ is a lower block diagonal matrix with identities along its diagonal, and thus invertible. We want to show that the affine control policy \eqref{eq:k_law} with $\mathbf{K} = \boldsymbol{\Phi_u} \boldsymbol{\Phi_x}^{-1}$ and $\mathbf{u_s}= \boldsymbol{\phi_u} - \boldsymbol{\Phi_u} \boldsymbol{\Phi_x}^{-1}\boldsymbol{\phi_x}$ achieve the desired system and affine system responses. We start by showing the linear portion of the parameterization, $\boldsymbol{\Phi_x}$ and $\boldsymbol{\Phi_u}$, which is analogous to the standard SLS parameterization \cite{sls}. Specifically, by inserting the expression $\mathbf{K} = \boldsymbol{\Phi_u} \boldsymbol{\Phi_x}^{-1}$ in expressions \eqref{eqn:Phi_x} and \eqref{eqn:Phi_u} we have:
\begin{equation}\label{eq:proof_sys_response}
\begin{aligned}
    &(I -\mathcal{Z}\mathcal{A} - \mathcal{Z}\mathcal{B}{\bf K})^{-1}\!\!=\!(I -\mathcal{Z}\mathcal{A} - \mathcal{Z}\mathcal{B}\boldsymbol{\Phi_u} \boldsymbol{\Phi_x}^{-1})^{-1} = \\
    & (((I-\mathcal{Z}\mathcal{A})\boldsymbol{\Phi_x} - \mathcal{Z}\mathcal{B}\boldsymbol{\Phi_u})\boldsymbol{\Phi_x}^{-1})^{-1} = (\boldsymbol{\Phi_x}^{-1})^{-1}= \boldsymbol{\Phi_x}, \notag
\end{aligned}
\end{equation}
and
\begin{equation*}
    \mathbf{K}\boldsymbol{\Phi_x} = \boldsymbol{\Phi_u} \boldsymbol{\Phi_x}^{-1}\boldsymbol{\Phi_x} = \boldsymbol{\Phi_u},
\end{equation*}
where we made use of equality \eqref{sys_constraints}. 

For the affine system response, we use the expression $\mathbf{u_s} = \boldsymbol{\phi_u} - \boldsymbol{\Phi_u\Phi_x}^{-1}\boldsymbol{\phi_x}$ in equation \eqref{eqn:phi_x} we get
\begin{equation*}
\begin{aligned}
    &\boldsymbol{\Phi_x}\mathcal{Z}(\mathcal{B}(\boldsymbol{\phi_u} - \boldsymbol{\Phi_u} \boldsymbol{\Phi_x}^{-1}\boldsymbol{\phi_x}) + \mathbf{s})\\
    &=\boldsymbol{\Phi_x}(\mathcal{Z}\mathcal{B}\boldsymbol{\phi_u} - \mathcal{Z}\mathcal{B}\boldsymbol{\Phi_u} \boldsymbol{\Phi_x}^{-1}\boldsymbol{\phi_x} + \mathcal{Z}\mathbf{s}),
\end{aligned}
\end{equation*}
Now we note from \eqref{sys_constraints} that $\mathcal{Z}\mathcal{B}\boldsymbol{\phi_u} = (I -\mathcal{Z}\mathcal{A})\boldsymbol{\phi_x} - \mathcal{Z}\mathbf{s}$ and $\mathcal{Z}\mathcal{B}\boldsymbol{\Phi_u}=(I-\mathcal{Z}\mathcal{A})\boldsymbol{\Phi_x} - I$, thus we can write
\begin{equation*}
\begin{aligned}
    &= \boldsymbol{\Phi_x}((I -\mathcal{Z}\mathcal{A})\boldsymbol{\phi_x} - ((I-\mathcal{Z}\mathcal{A})\boldsymbol{\Phi_x} - I) \boldsymbol{\Phi_x}^{-1}\boldsymbol{\phi_x}) \\
    &= \boldsymbol{\Phi_x}((I -\mathcal{Z}\mathcal{A})\boldsymbol{\phi_x} - (I-\mathcal{Z}\mathcal{A})\boldsymbol{\phi_x} + \boldsymbol{\Phi_x}^{-1}\boldsymbol{\phi_x}) \\
    &=\boldsymbol{\Phi_x}\boldsymbol{\Phi_x}^{-1}\boldsymbol{\phi_x} =\boldsymbol{\phi_x}.
\end{aligned}
\end{equation*}
Finally, concerning $\boldsymbol{\phi_u}$, we note from \eqref{eqn:Phi_u} and \eqref{eqn:phi_x} that $\boldsymbol{\Phi_u}\mathcal{Z}(\mathcal{B}\mathbf{u_s}+\mathbf{s})=\mathbf{K}\boldsymbol{\Phi_x}\mathcal{Z}(\mathcal{B}\mathbf{u_s}+\mathbf{s}) = \mathbf{K}\boldsymbol{\phi_x}$. Thus in  \eqref{eqn:phi_u} we get
\begin{equation*}
    \mathbf{K}\boldsymbol{\phi_x} + \mathbf{u_s} = \boldsymbol{\Phi_u} \boldsymbol{\Phi_x}^{-1}\boldsymbol{\phi_x} + \boldsymbol{\phi_u} -  \boldsymbol{\Phi_u} \boldsymbol{\Phi_x}^{-1}\boldsymbol{\phi_x} = \boldsymbol{\phi_u}.
\end{equation*}
\end{proof}

\begin{remark}
Note that the affine SLS parameterization is a strict generalization of the standard SLS parameterization \cite{sls}. In particular, it allows for the parameterization of affine polices even if the system itself is linear, namely $\mathbf{s} = \mathbf{0}_{n(T+1)}$. This comes from the fact that our approach parameterizes one additional column of the system response that is mapped to $\mathcal{Z}\mathbf{s}$ for affine systems or to $\mathbf{0}_{n(T+1)}$ for linear systems. As previously discussed, this enables to capture affine closed-loop behaviors, emerging from MPC policies. 
\end{remark}

Note that, in the absence of noise, i.e., when $\mathbf{w}=\begin{bmatrix}x_0^\intercal & \mathbf{0}_n^\intercal & \dots & \mathbf{0}_n^\intercal \end{bmatrix}^\intercal$, \Cref{theorem_sls} can be simplified.
In that case, only the first $n$ columns of $\boldsymbol{\Phi_x}$ and $\boldsymbol{\Phi_u}$ that map the initial condition to the state and input trajectory, respectively, have to be parameterized.

\begin{corollary}[SLS with Affine Policies for Noiseless Systems]
    Consider the dynamics of system \eqref{eqn:system} with $w(t)=0$, $\forall t >0$. Let $\boldsymbol{\Phi_\ell}\{0\}$, $\boldsymbol{\ell} \in \{\mathbf{x},\mathbf{u}\}$ denote the first $n$ columns of the system response \eqref{closed_loop_response} and $\mathbf{s}=[s^\top,...,s^\top]^\top\in\mathbb{R}^{nT}$. Subspace \eqref{sys_constraints} can be simplified to
\begin{equation} \label{noiseless_constr}
    \begin{aligned}
    &\begin{bmatrix} I_{n(T+1)}-\mathcal{Z}\mathcal{A} & -\mathcal{Z} \mathcal{B} \end{bmatrix}
    \begin{bmatrix}  \boldsymbol{\Phi_{x}}\{0\} & \boldsymbol{\phi_x} \\   \boldsymbol{\Phi_{u}}\{0\} &\boldsymbol{\phi_u}\end{bmatrix}\!\!=\!\!\begin{bmatrix} I_{n} &0 \\0 &  \mathbf{s}\end{bmatrix}.
     \end{aligned}
\end{equation}
The closed-loop map from $x_0$ to the state and input trajectory is then given by
\begin{equation}\label{eqn:noiseless_affine_SLS}
\begin{bmatrix} \mathbf{x} \\ \mathbf{u} \end{bmatrix}  = \begin{bmatrix}  \boldsymbol{\Phi_{x}}\{0\} & \boldsymbol{\phi_x} \\   \boldsymbol{\Phi_{u}}\{0\} &\boldsymbol{\phi_u}\end{bmatrix}  \begin{bmatrix} x_0 \\ \mathbf 1\end{bmatrix} =: \begin{bmatrix}\boldsymbol{\bar\Phi_x} \\ \boldsymbol{\bar\Phi_u}\end{bmatrix} \begin{bmatrix} x_0 \\ \mathbf 1\end{bmatrix}.
\end{equation}
\end{corollary}

\begin{proof}
    The proof follows from \Cref{theorem_sls} with $\mathbf{w}=\begin{bmatrix}x_0^\intercal & 0_n^\intercal & \dots & 0_n^\intercal \end{bmatrix}^\intercal$.
\end{proof}

\subsection{System Level Synthesis with Affine Policies for MPC}
Here, we leverage the affine SLS parameterization for time-varying affine policies to establish an explicit formulation of the solution to the MPC problem \eqref{Traditional_MPC}. We make use of the fact that control policies resulting from a quadratic MPC problem with linear constraints as in \eqref{Traditional_MPC} are piecewise affine \cite{borrelli_constrained}, and demonstrate that our novel SLS formulation with affine policies can be used to equivalently parameterize \eqref{Traditional_MPC} with closed-loop maps.

\begin{lemma}[Explicit solutions to MPC]\label{lemma:affine_policies_opt_control}
The optimal control input $ u^*_{t} $ obtained as a solution to the MPC problem \eqref{Traditional_MPC} with cost function \eqref{cost_fn} lies within the affine subspace spanned by the corresponding optimal state $x^*_{t}$. Namely, for all $t \in [0,T] $, there exist coefficients $K_t \in \mathbb{R}^{m \times n}$ and $ u_{t,s} \in \mathbb{R}^m$ such that  $u^*_{t} = K_{t} x^*_{t} + u_{t,s}.$
\end{lemma}
\begin{proof}
It is a well-known fact that the solution \( u_{0}^* \) to problem \eqref{Traditional_MPC} is a piecewise affine map. Specifically, it is piecewise affine over polyhedral regions \( \mathcal{X}_{\text{poly}} \subseteq \mathcal{X}_L \), where \( \mathcal{X}_L \) is the feasible set of problem \eqref{Traditional_MPC}. In particular, by \cite[Theorem 8.1/8.3]{borrelli_constrained}, in each polyhedral region \( \mathcal{X}_{\text{poly}} \) there exist coefficients \( K_{0} \) and \( u_{0,s} \) such that the optimal control input \( u_0^* \) can be expressed as
$u_0^* = K_{0} x^*_0 + u_{0,s}.$
We use this fact to recursively show that each policy at time \( t \in [0,T] \) is also affine. In particular, since \( u_t^* \) is the solution to the same optimization problem but for a reduced time horizon \( T-t \) with initial condition \( x_t^* \), there exist corresponding coefficients \( K_t \) and \( u_{t,s} \) such that $u_t^* = K_t x_t^* + u_{t,s}$.
This holds true for every \( t \in [0,T] \), where problem \eqref{Traditional_MPC} has a time horizon of $T-t$. Thus, all solutions \( u^*_{0:T} \) are spanned by the affine functions \( K_{t} x_t + u_{t,s} \) for each time step \( t \in [0,T] \).
\end{proof}

Using \Cref{lemma:affine_policies_opt_control}, we can formulate the following equivalent characterization of the MPC problem \eqref{Traditional_MPC} with cost \eqref{cost_fn}.
\begin{corollary}\label{cor:affine_sls_CFTOC}
    The MPC problem \eqref{Traditional_MPC} and the optimization problem
    \begin{align}\label{eq:affine_sls_model_based}
    &\underset{\boldsymbol{\bar\Phi_x, \bar\Phi_u}}{\min}&  & \quad J(\boldsymbol{\bar\Phi_x}  \begin{bmatrix} x_0 \\ \mathbf 1\end{bmatrix}, \boldsymbol{\bar\Phi_u}  \begin{bmatrix} x_0 \\ \mathbf 1\end{bmatrix})\\
    & \text{s.t.} &  &\begin{aligned} \nonumber
    & \begin{bmatrix} I-\mathcal{Z}\mathcal{A} & -\mathcal{Z} \mathcal{B} \end{bmatrix}
    \begin{bmatrix} \boldsymbol{\bar\Phi_x} \\ \boldsymbol{\bar\Phi_u}\end{bmatrix} = \begin{bmatrix} I_{n} &0 \\0 &  \mathbf{s}\end{bmatrix} \\
    & \mathcal{H}_x \boldsymbol{\bar\Phi_x}  \begin{bmatrix} x_0 \\ \mathbf 1\end{bmatrix}+ \mathcal{H}_u \boldsymbol{\bar\Phi_u}  \begin{bmatrix} x_0 \\ \mathbf 1\end{bmatrix} \leq \mathbf h\\
    &x_0 = x(0), \text{ }x_T \in \mathcal{X}_T
\end{aligned}
\end{align}
with cost function $J$ from definition \eqref{cost_fn} and constraints $\mathcal H_x = diag(H_x)$, $\mathcal H_u=diag(H_u)$ adapted to signal notation, give rise to the same optimal trajectories $\{x^*_{0:t},u^*_{0:T}\}$. 
That is,  $\{x^*_{0:t},u^*_{0:T}\}$ is an optimal solution to \eqref{Traditional_MPC} if and only if it satisfies \eqref{eqn:noiseless_affine_SLS} for $\boldsymbol{\bar\Phi_x}$ and $\boldsymbol{\bar\Phi_u}$ being the optimal solutions to~\eqref{eq:affine_sls_model_based}.

\end{corollary}
\begin{proof}

By Lemma \ref{lemma:affine_policies_opt_control}, all optimal solution to MPC problem \eqref{Traditional_MPC} are given as $u^*_t=K_tx^*_t+u_{s,s} \quad \forall t\in[0,T]$. Hence, they are always attainable by \eqref{eq:affine_sls_model_based}, which optimizes over all closed-loop system responses with time-varying affine policies (as defined in \eqref{ctrl_policy}). Moreover, by Corollary \ref{theorem_sls}, closed-loop time-varying affine system responses $\boldsymbol{\bar\Phi_x}$ and  $\boldsymbol{\bar\Phi_u}$ parameterize all possible system responses of system \eqref{eqn:system} with affine control policies \eqref{ctrl_policy}. Therefore, formulations \eqref{Traditional_MPC} and \eqref{eq:affine_sls_model_based} are equivalent via \eqref{eqn:noiseless_affine_SLS}.

\end{proof}

\section{Data-Driven System Level Synthesis with Affine Policies} \label{sec:data_driven_appr}
Here, we consider the case where dynamic matrices $A$, $B$ and the affine part $s$ are unknown but offline data is available. In this case, we assume that dynamics \eqref{eqn:system} are time-invariant, and show that all affine system responses \eqref{noiseless_constr} can be equivalently characterized by a set of past closed-loop trajectories in a data-driven fashion. We show how the resulting data-driven closed-loop system response can be used to equivalently parameterize problem \eqref{Traditional_MPC}.

In what follows, we introduce some known concepts that will be useful in the remainder of the section.

\begin{definition}[Hankel Matrix \cite{Hankel_PE}]
    The Hankel matrix of depth $L$ of a signal $\sigma: \mathbb{Z} \to \mathbb{R}^p$ of length $T$, is defined as
\begin{equation*}
   H_{L}(\sigma(0:T-1)) = \begin{bmatrix} \sigma(0) &  \sigma(1) & \hdots &   \sigma(T-L) \\
    \sigma(1) &  \sigma(2)  & \hdots & \sigma(T-L+1) \\
    \vdots & \vdots & \ddots & \vdots \\
     \sigma(L-1) & \sigma(L) & \hdots &  \sigma(T-1) \end{bmatrix}.
\end{equation*}
\end{definition}
Here the notion of persistency of excitation is central for Willem's lemma \cite{willems} and its extension to affine systems  \cite{willem_affine}.
\begin{definition}[Persistent Exciting Signal \cite{Hankel_PE}]
Let $ \sigma: \mathbb{Z} \rightarrow \mathbb{R}^p$ be a signal. We say that its finite time restriction $ \tilde{\sigma}:[0,T-1]\to\mathbb{R}^p $ is persistently exciting (PE) of order L if the Hankel matrix $H_{L}( \tilde{\sigma}(0:T-1))$ has full rank.
\end{definition}
 Building on these concepts one can extend Willem's fundamental lemma \cite{willems}, to affine systems\cite{Martinelli_2022} \cite{willem_affine}.
\begin{lemma}[Willem's lemma for affine systems \cite{Martinelli_2022}\label{willems_thm}] 
Consider the discrete-time invariant affine system from equation \eqref{eqn:system} without noise, i.e $w(t)=0$, $\forall t\geq 0$. Assume the pair $(A,B)$ to be controllable. Let $ \{ \tilde u(0:T-1), \tilde x(0:T-1) \}$ be a set of state and input signals generated by the system.
If  $\tilde u(0:T-1)$ is persistently exciting of order $n+L+1$, the pair $ ( u(0:L-1)$, $x(0:L-1))$ is a trajectory of the system if and only if there exists $g \in \mathbb{R}^{T-L+1}$ such that 
\begin{equation*}
    \begin{bmatrix} u(0:L-1) \\ x(0:L-1) \\ 1 \end{bmatrix} = \begin{bmatrix} H_L(\tilde u(0:T-1)) \\ H_L( \tilde x(0:T-1)) \\ \mathbf{1}_{T-L+1}^\top\end{bmatrix}  g.
\end{equation*}
\end{lemma}

\subsection{Data-Driven System Level Synthesis with Affine Policies}
We now show that all affine system responses can be parameterized as a function of the Hankel matrices, a matrix $G$ and a vector $\hat g$.
\begin{theorem}(Data Driven SLS with Affine Policies) \label{main_thm}
Consider the discrete-time invariant affine system with dynamics \eqref{eqn:system},  $w(t)=0$ $\forall t\geq 0$, and (A,B) controllable. Suppose that a state/input pair $(\mathbf{\tilde x}, \mathbf{\tilde u}) \in \mathbb{R}^{nL}\times\mathbb{R}^{mL}$ is collected with input $\mathbf{\tilde u}$ PE of order $n + (T+1) + 1$, the system responses $\{ \boldsymbol{\bar \Phi_x}, \boldsymbol{\bar \Phi_u} \}$ satisfying parameterization \eqref{noiseless_constr} over a time horizon $t \in [0,T]$ can be equivalently characterized as
\begin{align*}
 &\left\{ \begin{bmatrix} H_{T+1}(\mathbf{\tilde x}) \\ H_{T+1} \nonumber (\mathbf{\tilde u}) \end{bmatrix} \begin{bmatrix} G - \hat{g} \mathbf{1}_n^\top & \hat g\end{bmatrix}, \quad \forall G, g \in \Gamma(\mathbf{\tilde x})\times \Lambda(\mathbf{\tilde x}) \right\}\\
& \begin{aligned} \nonumber
& \Gamma(\mathbf{\tilde x}): = \{G: H_1(\mathbf{\tilde x})  G = I_n \text{, } \mathbf{1}_{L-T}^\top  G = \mathbf{1}_n^\top \}\\
&\Lambda(\mathbf{\tilde x}):= \{\hat g:H_1(\mathbf{\tilde x})  \hat g = \mathbf{0}_n \text{, } \mathbf{1}_{L-T}^\top \hat g = 1\}
\end{aligned}
\end{align*}
with $G \in \mathbb{R}^{L-T \times n}$, $\hat g \in \mathbb{R}^{L-T}$ and the symbol $H_1$ represents the first $n$ rows of the Hankel matrix $ H_{T+1}(\mathbf{\tilde x})$. 
\end{theorem}
\begin{proof}
$(\subseteq)$
Consider a set of system responses $\{ \boldsymbol{\bar \Phi_x}, \boldsymbol{\bar \Phi_u} \}$ that satisfies  equation \eqref{noiseless_constr}. By \Cref{willems_thm} it follows that, for all initial conditions $x_0$, a vector $g\in \mathbb{R}^{L-T}$ exists such that
\begin{align*}
 & \begin{bmatrix} \mathbf{x} \\ \mathbf{u} \end{bmatrix}=\begin{bmatrix}  \boldsymbol{\Phi_{x}}\{0\} & \boldsymbol{\phi_x} \\   \boldsymbol{\Phi_{u}}\{0\} &\boldsymbol{\phi_u}\end{bmatrix} \begin{bmatrix}  x_0 \\ 1\end{bmatrix}  = \begin{bmatrix} H_{T+1}(\mathbf{\tilde x}) \\ H_{T+1}(\mathbf{\tilde u}) \end{bmatrix} \cdot g \\
& \text{s.t.} \quad  \begin{aligned} 
&\mathbf{1}_{L-T}^\top g = 1. \\
\end{aligned}
\end{align*}
To find a parameterization of the last column of the affine system response, or equivalently for $x_0 = \mathbf{0}_{n}$, we choose $\hat g$, such that $\mathbf{1}_{L-T}^\top  \hat g = 1$. Note that this leads to the additional constraint $H_1(\mathbf{\tilde x}) \cdot \hat g = 0_n$, since the first set of n-rows of $\bar \Phi_x$ need to be zero as they map the constant term to the initial condition $x_0= {\bf 0}_n$.
To find a parameterization of the first $n$-columns of the system response, we choose $x_0 = e^{(0)}$, the first unit vector in $\mathbb{R}^n$, which will give us
\begin{align*}
&  \begin{bmatrix}  \boldsymbol{\Phi_{x}}\{0\} & \boldsymbol{\phi_x} \\   \boldsymbol{\Phi_{u}}\{0\} &\boldsymbol{\phi_u}\end{bmatrix}  \begin{bmatrix} e^{(0)} \\ 1 \end{bmatrix}  = \begin{bmatrix} H_{T+1}(\mathbf{\tilde x}) \\ H_{T+1}(\mathbf{\tilde u}) \end{bmatrix} g^{(0)} \\
& \text{s.t} \quad \begin{aligned}
& \mathbf{1}_{L-T}^\top  g^{(0)} = 1\\
\end{aligned}
\end{align*}
We can repeat this process $n$ times for each unit vector $e^{(0)}, e^{(1)}, ..., e^{(n)}$ and the corresponding vectors $G = [g^{(0)}, g^{(1)}, ..., g^{(n)}]\in \mathbb{R}^{(L-T)\times n }$, to end up with 
\begin{align*}
  & \begin{bmatrix}  \boldsymbol{\Phi_{x}}\{0\} & \boldsymbol{\phi_x} \\   \boldsymbol{\Phi_{u}}\{0\} &\boldsymbol{\phi_u}\end{bmatrix}\begin{bmatrix} I_n\\ \mathbf{1}_n^\top \end{bmatrix} 
 =  \begin{bmatrix} H_{T+1}(\mathbf{\tilde x}) \\ H_{T+1}(\mathbf{\tilde x}) \end{bmatrix} G\\
&\text{s.t.} \quad \begin{aligned}
&\mathbf{1}_{L-T}^\top G = \mathbf{1}_n^\top.
\end{aligned}
\end{align*}
Now, to get the parameterization of the first $n$ columns of $\boldsymbol {\bar \Phi_x}$ and $\boldsymbol {\bar \Phi_u}$ we subtract the above parameterization of the last column after concatenating it $n$ times, leaving us with
\begin{align*}
     &\begin{bmatrix}  \boldsymbol{\Phi_{x}}\{0\} & \boldsymbol{\phi_x} \\   \boldsymbol{\Phi_{u}}\{0\} &\boldsymbol{\phi_u}\end{bmatrix} \begin{bmatrix} I_n \\ \mathbf{0}_n^\top \end{bmatrix}= \begin{bmatrix} H_{T+1}(\mathbf{\tilde x}) \\ H_{T+1}(\mathbf{\tilde u}) \end{bmatrix} \begin{bmatrix} G - \hat g \mathbf{1}_{n}^\top\end{bmatrix}\\
    &\text{s.t.} \quad \begin{aligned}
    & \mathbf{1}_{L-T}^\top  G = \mathbf{1}_n^\top \text{ and } \mathbf{1}_{L-T}^\top \hat g = 1,
\end{aligned}
\end{align*}
thus proving the first direction. Note that this leads to the additional constraint $H_1 \cdot G = I_n$, because the first $n$ rows of $\bar \Phi_x$ map the initial condition $x_0$ to itself and $H_1(\tilde{\mathbf{x}})\hat g=0$.\newline
$(\supseteq)$ First we note that each column of the Hankel matrix $[H_{T+1}]_{:,i}$, $\forall i \in [0,L-T]$  satisfies the dynamics
\begin{align*}
    [H_{T+1}(\mathbf{\tilde x})]_{:,i} =&
    \mathcal{Z}\mathcal{A}[H_{T+1}(\mathbf{\tilde x})]_{:,i}\\ + &\mathcal{Z}\mathcal{B}[H_{T+1}(\mathbf{\tilde u})]_{:,i} + \begin{bmatrix} [H_{T+1}(\mathbf{\tilde x})]_{:,i} \\ \mathbf{s}\end{bmatrix}
\end{align*}
By concatenating each column we get
\begin{equation*}
    \begin{bmatrix} I_{n(T+1)}-\mathcal{Z}\mathcal{A} & -\mathcal{Z}\mathcal{B} \end{bmatrix}  \begin{bmatrix} H_{T+1}(\mathbf{\tilde x}) \\ H_{T+1}(\mathbf{\tilde u}) \end{bmatrix} = \begin{bmatrix} H_1(\mathbf{\tilde x}) \\ \mathbf{s} \mathbf{1}_{L-T}^\top \end{bmatrix}.
\end{equation*}
Upon multiplying the matrix $\begin{bmatrix} G - \hat{g}\mathbf{1}_{n} & \hat g \end{bmatrix}$, with $G$ and $\hat g$ satisfying the  conditions in \Cref{main_thm}, for the left side of the previous equation we obtain
\begin{equation*}
\begin{bmatrix} I_{n(T+1)}-\mathcal{Z}\mathcal{A} & -\mathcal{Z} \mathcal{B} \end{bmatrix}
     \begin{bmatrix}  \boldsymbol{\Phi_{x}}\{0\} & \boldsymbol{\phi_x} \\   \boldsymbol{\Phi_{u}}\{0\} &\boldsymbol{\phi_u}\end{bmatrix} 
\end{equation*}
and for the right-hand side
\begin{align*}
    &\begin{bmatrix} H_1(\mathbf{\tilde x}) \\ \mathbf{s}\mathbf{1}_{L-T}^\top\end{bmatrix}
    \begin{bmatrix} G - \hat g \mathbf{1}_n^\top & \hat g \end{bmatrix}=\\
    &\begin{bmatrix} I_n & \mathbf{0}_n \\ \mathbf{s}\mathbf{1}_{L-T}^\top \begin{bmatrix} G - \hat g \mathbf{1}_n^\top \end{bmatrix} & \mathbf{s}\mathbf{1}_{L-T}^\top \cdot \hat g\end{bmatrix} 
    = \begin{bmatrix} I_n & 0 \\ 0 & \mathbf{s}\end{bmatrix}.
\end{align*}
The first equality follows directly from $H_1(\mathbf{\tilde x})G=I_n, H_1(\mathbf{\tilde x})\hat g=\mathbf{0}_n$. 
and the second equality follows from $\mathbf1_{L-T}^\top  G = \mathbf 1_n^\top \text{, } \mathbf 1_{L-T}^\top \hat g = 1$, thus concluding the proof.
\end{proof}

\subsection{Data-Driven SLS with Affine Policies for MPC}
The data-driven SLS formulation provided in \Cref{main_thm} can be used to reformulate problem \eqref{Traditional_MPC} in a data-driven manner. In particular, given a collected trajectory $( \mathbf{\tilde u}, \mathbf{\tilde x})$, with $\mathbf{\tilde u}$ persistently exciting of order $n + (T+1) + 1$, we can reformulate the MPC problem \eqref{cor:affine_sls_CFTOC} as:

\begin{corollary}\label{cor:data_affine_sls_CFTOC}
Under the conditions of \Cref{main_thm}, the MPC subroutine with time horizon $T$
\begin{align} \label{eq:dd_affine_sls}
    &\min_{\mathbf x, \mathbf u, G, \hat g}  \quad\sum_{t=0}^{T}\quad J(\mathbf x, \mathbf u)\\
    &\begin{aligned}
    \text{s.t.} & \begin{bmatrix} \mathbf x \\ \mathbf u \end{bmatrix} = \begin{bmatrix} H_{T+1}(\mathbf{\tilde x}) \\ H_{T+1}(\mathbf{\tilde u}) \end{bmatrix} \begin{bmatrix} G - \hat g \mathbf{1}_n^\top & \hat g\end{bmatrix}  \begin{bmatrix} x_0 \\  \mathbf 1\end{bmatrix}\\
    & G,\hat g \in \Gamma ( \tilde{\mathbf{x}} ) \times \Lambda ( \mathbf{\tilde{x}}) \nonumber  \\
    & \mathcal H_x \mathbf x+ \mathcal H_u \mathbf u \leq \mathbf h\\
    & x_0 = x(0), \text{ }x_T \in \mathcal{X}_T
\end{aligned}
\end{align}
with $J(\mathbf x,\mathbf u)$ from definition \eqref{cost_fn} and constraints $\mathcal H_x = \rm{diag}(H_x)$ and $\mathcal H_u=\rm{diag}(H_u)$ adapted to  signal notation, is solved by the same optimal solution as the MPC problem \eqref{Traditional_MPC}. That is, the sequence $\{x^*_{0:t},u^*_{0:T}\}$ is an optimal solution to problem \eqref{Traditional_MPC} if and only if it is an optimal solution to problem \eqref{eq:dd_affine_sls}.
\end{corollary}
\begin{proof}
By \Cref{main_thm}, problem \eqref{eq:dd_affine_sls} is equivalent to \eqref{eq:affine_sls_model_based}. Consequently, the proof follows directly from \Cref{cor:affine_sls_CFTOC}.
\end{proof}

\begin{remark}
    Note that, in formulation \eqref{eq:dd_affine_sls} dynamic matrices $A$ and $B$, as well as the constant $\mathbf{s}$ are learned from data. This enables data-driven computation of closed-loop maps for affine systems without requiring prior knowledge of $\mathbf{s}$.
\end{remark}

\section{Simulation}
\label{section_simulation}
We demonstrate the effectiveness of our proposed algorithms, given in formulations \eqref{eq:affine_sls_model_based} and \eqref{eq:dd_affine_sls}, by comparing them to the traditional nominal formulation in problem \eqref{Traditional_MPC}. To this end, we consider the linearized swing dynamics with constant Coulomb friction, which constitutes a two-dimensional affine system with states $\theta$ (phase angle deviation) and $\omega$ (frequency deviation), and $u$ representing the control action. This affine system model captures the essential behavior of power system generators under friction effects: \footnote{Code implementation of the described simulation can be found at: \url{https://github.com/lukaschu/SLS-for-Affine-Policies}.}
\begin{equation} \label{eq:sys_dyn}
\begin{bmatrix}
\theta(t+1) \\ \omega(t+1)
\end{bmatrix}
= A \begin{bmatrix}
\theta(t) \\ \omega(t)
\end{bmatrix} + B u(t) + s .
\end{equation}
The system matrices $\{A,B\}$ and the constant $s$ are given as
\begin{equation} 
A=\!\!=\!\!\begin{bmatrix} 1 & \Delta t \\ -\cfrac{\Delta t}{m} & 1\!-\!\cfrac{d\Delta t}{m} \end{bmatrix},B=\!\!\begin{bmatrix} 0 \\ 1 \end{bmatrix},s=\!\!\!\begin{bmatrix} 0 \\ c \end{bmatrix}.
\end{equation} 
The parameters are chosen as follows: the inverse inertia coefficient $m^{-1}=0.8$, the damping coefficient $d=0.5$, and the constant Coulomb friction $c = 0.1$. We choose a discretization timestep of $\Delta t=0.2$.

We define a time-invariant quadratic cost function with a reference point at \( \theta_{\text{ref}} = 0.6 \) given by  
\begin{equation}
J(\theta(t), \omega(t)) = (\theta(t) - \theta_{\text{ref}})^2
\end{equation}  

The system is subject to the following time-invariant state and input constraints:  
\begin{equation}
\begin{aligned}
    \mathcal{X} &:= \{\theta \mid -1.2  \leq \theta \leq 1.2, \quad -0.8  \leq \omega \leq 0.8 \}, \\
    \mathcal{U} &:= \{u \mid -0.5  \leq u \leq 0.5 \}.
\end{aligned}
\end{equation}  

For the data-driven approach, we assume access to a set of observed system behaviors \( \{\tilde{x}, \tilde{u} \} \) from \eqref{eq:sys_dyn}, where \( \tilde{u} \) is PE of order \( 2 + (T+1) + 1 \). The results for the initial condition \( \theta(0) = 0 \), \( \omega(0) = 0 \) and a finite horizon $T=10$ over a time horizon of $15$ steps are presented in \Cref{fig:mpc_comparison}, comparing all three controllers. One can see that both the model-based SLS parameterization and the data-driven SLS formulation lead to the same closed-loop solutions as the traditional MPC.

\begin{figure}[!t] 
    \centering
    \includegraphics[width=0.95\linewidth]{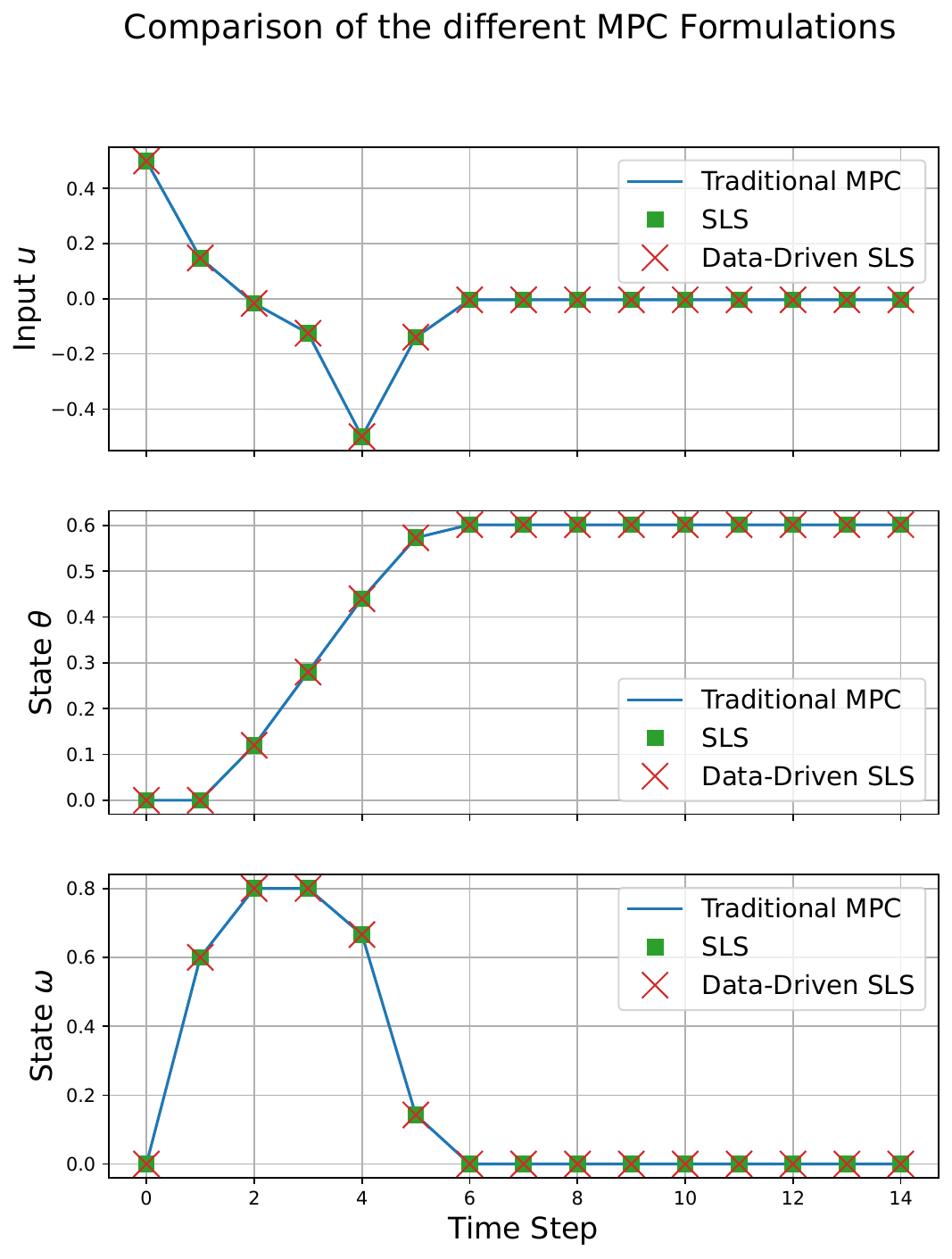}
    \caption{Comparison of the state and input trajectories obtained from the traditional MPC, SLS and data-driven SLS formulations.}
    \label{fig:mpc_comparison}
\end{figure}

\section{Conclusions}
\label{discussion_section}
We introduced the affine SLS approach, which provides a parameterization of all closed-loop system responses under time-varying affine policies. We demonstrated that any nominal MPC formulation with quadratic cost, affine constraints and affine systems, can be equivalently expressed using this new closed-loop representation. Furthermore, we extended our parameterization to the data-driven case, and established that all affine system responses can be characterized purely in terms of observed closed-loop system behaviors. We then showed how this characterization can be integrated into a nominal data-driven MPC framework. Finally, we validated our approach through simulations, demonstrating that our method yields solutions equivalent to those obtained using traditional model-based MPC. An extension of this approach to the distributed case \cite{dlmpc} providing guarantees on asymptotic stability and recursive feasibility \cite{dlmpc2} is deferred to future work.

\bibliographystyle{IEEEtran} 
\bibliography{ref}
\end{document}